\documentclass{paper}
\usepackage{amsmath}
\usepackage{amsthm}
\usepackage{amssymb}
\usepackage{framed}
\usepackage[margin=1in]{geometry}

\usepackage[utf8x,utf8]{inputenc} 
\usepackage[T1]{fontenc}
\usepackage{graphicx}
\usepackage{todonotes}
\usepackage{mathrsfs}
\usepackage{hyperref}
\usepackage{enumitem}

\setlist[enumerate,1]{%
  label=\arabic*.,
}

\newlist{inlinelist}{enumerate*}{1}
\setlist*[inlinelist,1]{%
  label=(\roman*),
}

\usepackage[noend]{algorithmic}
\usepackage[boxed]{algorithm}

\newcommand{\LineIfElse}[3]{ \STATE \algorithmicif\ {#1}\ \algorithmicthen\ {#2} \algorithmicelse\ {#3} }

\newcommand{\algorithmiccommentt}[1]{\colorbox{black!10}{#1}}

\newtheorem{theorem}{Theorem}

\newtheorem{claim}[theorem]{Claim}

\newcommand{\subsetsum}{\textsc{Subset Sum}}

\numberwithin{equation}{section}

\newtheorem*{rep@theorem}{\rep@title}
\newcommand{\newreptheorem}[2]{%
\newenvironment{rep#1}[1]{%
 \def\rep@title{#2 \ref{##1}}%
 \begin{rep@theorem}[restated]}%
 {\end{rep@theorem}}}

\makeatother

\newreptheorem{theorem}{Theorem}
\newreptheorem{lemma}{Lemma}
\newreptheorem{proposition}{Proposition}

\begin{document}
\title{A short note on Merlin-Arthur protocols for subset sum}
\author{Jesper Nederlof\thanks{Department of Mathematics and Computer Science, Technical University of Eindhoven, The Netherlands, \texttt{j.nederlof@tue.nl}}}
\maketitle

In the \subsetsum{} problem we are given positive integers $w_1,\ldots,w_n$ along with a target integer $t$. The task is to determine whether there exists a subset $X \subseteq \{1,\ldots,n\}$ such that $\sum_{i \in X}w_i=t$. We refer to such an $X$ as a \emph{solution}. The goal of this short note is to prove the following:

\begin{theorem}\label{mainthm}
For a given \subsetsum{} instance $w_1,\ldots,w_n,t$ there is a proof of size\footnote{$O^*(\cdot)$ omits factors polynomial in the input size.} $O^*(\sqrt{t})$ of what the number of solutions is that can be constructed in $O^*(t)$ time and can be probabilistically verified in time $O^*(\sqrt{t})$ with at most constant error probability.
\end{theorem}

In complexity theory, a proof system as above is commonly referred to as a \emph{Merlin-Arthur protocol}. These protocols very recently received attention~\cite{2016arXiv160104743W, 2016arXiv160201295B} in the exponential time setting: Williams~\cite{2016arXiv160104743W} gives very generic Merlin-Arthur protocols with verifiers more efficient than best known exponential time algorithms.
By no means we claim this note is innovative in any way: both our work and~\cite{2016arXiv160104743W} have a similar flavor and it is likely that Theorem~\ref{mainthm} is also obtainable via a clever application of~\cite[Theorem 1.1]{2016arXiv160104743W}, or at the very least our observation will be less surprising given~\cite{2016arXiv160104743W, 2016arXiv160201295B}.

\begin{proof}[Proof of Theorem~\ref{mainthm}]
The proof system is outlined in Algorithm~\ref{alg:proof}. Note that both the prover and the verifier algorithms are simple modifications of the well known pseudo-polynomial time dynamic programming algorithm for \subsetsum{} usually attributed to Bellman.
\begin{algorithm}
  \caption{The prover-verifier protocol.}
  \label{alg:proof}
  \begin{algorithmic}[1]
    \REQUIRE $\mathsf{P}(w_1,\ldots,w_n,t)$\hfill\algorithmiccommentt{Prove that the number of solutions is $c_t$.}
    \ENSURE Prime $p=\Theta(\sqrt{nt})$, integer $c_i$ for $i \leq nt,i \equiv_p t$ with $c_i = |\{ X \subseteq [n]: w(X)=i\}|$.
		\STATE Initiate $T[0,0]=1$ and $T[0,i]=0$ for $0<i\leq nt$.  
		\FOR{$j=1 \to n$}
			\FOR{$i=1 \to n\cdot t$}
			\LineIfElse{$i < w_j$}{$T[j,i] \gets T[j-1,i]$}{$T[j,i] \gets T[j-1,i] + T[j-1,i-w_j]$.}
			\ENDFOR
    \ENDFOR
	\STATE Pick smallest prime $p$ satisfying $2\sqrt{nt}<p < 4\sqrt{nt}$.\hfill\algorithmiccommentt{Brute-force and~\cite{Agrawal02primesis} is within time budget.}\label{pp1}
	\FOR{$i \leq n\cdot t$ such that $i \equiv_p t$}
		\STATE $c_i \gets T[n,i]$.
	\ENDFOR
  \end{algorithmic}
  \begin{algorithmic}[1]
    \REQUIRE $\mathsf{V}(w_1,\ldots,w_n,t, p, \{c_i\})$\hfill\algorithmiccommentt{Verify the proof for the number of solutions.}
    \ENSURE Yes if the proof is as output by $\mathsf{P}$, no with constant probability if the proof is `incorrect'.
		\STATE Pick a prime $q$ satisfying $2^n t < q < 2^{n+1}t$. \hfill\algorithmiccommentt{Sample and check primality until success.}\label{pp2}
		\STATE Pick a random element $r \in \mathbb{Z}_{q}$.
		\STATE Initiate $T'[0,0]=1$ and $T'[0,i]=0$ for $i \in \mathbb{Z}_p$.
		\FOR{$j=1 \to n$}
			\FOR{$i=1 \to n\cdot t$}
				\STATE $T'[j,i] \gets \left(T'[j-1,i] + r^{w_j}\cdot T'[j-1,(i-w_j) \% p]\right) \% q$.\label{moda}
			\ENDFOR
    \ENDFOR
	\STATE Compute $\sum_{i}c_ir^i \% q$.
	\LineIfElse{$\sum_{i}c_ir^i \equiv_q T'[n,t \% p]$}{\algorithmicreturn\ yes}{\algorithmicreturn\ no.}
  \end{algorithmic}
\end{algorithm}

Recall the Prime Number Theorem, that is used in Lines~\ref{pp1} and~\ref{pp2}:

\begin{claim}
\label{lem:prime-num}
For large enough integers $b$ 
there exist at least $2^b/b$ prime numbers $p$ in the
interval $2^b<p<2^{b+1}$.
\end{claim}

By this result, we may sample $2^n t < q < 2^{n+1}t$ and check primality until success and declare YES after $n+\lg(t)$ unsuccessful tries (which happens with probability at most $1/4$).

For the verification probability, it is straightforward to see that $T[n,i]$ equals $|\{ X \subseteq [n]: w(X)=j\}|$. Similarly, when considering $r$ as indeterminate, we see that $T'[n,t]$ is a polynomial in $\mathbb{Z}_q[r]$ satisfying
\[
	T'[n,t] \equiv_q \sum_{\substack{i \leq nt \\ i \equiv_p t}}|\{ X \subseteq [n]: w(X)=i\}| r^i.
\]
Also note that $|\{ X \subseteq [n]: w(X)=j\}| < q$. Thus we see that if the proof is correctly constructed as in $\mathsf{P}$, $\sum_{i}c_ir^i \% q$ and $T'[n,t]$ are equivalent polynomials and $\mathsf{V}$ accepts. On the other hand, if the proof is incorrect, e.g. $c_i \neq T[n,i]$ for some $i$, we see that $(\sum_{i}c_ir^i \% q) - T'[n,t]$ is not the zero polynomial and has degree at most $n\cdot t$. Thus, it has at most $n\cdot t$ roots in $\mathbb{Z}_q$, and the probability that $r$ is one of these roots is $nt/q < 1/4$. If $r$ is not a root of the difference polynomial, the two evaluations necessarily differ and the verifier rejects. Thus in total the verifier does not reject an incorrect proof with probability at most $1/2$.

For the running time of $\mathsf{P}$, note that all involved integers are at most $2^n$ and the running time bound follows directly. The running time of $\mathsf{V}$ follows similarly by using taking modulus at each powering step at Line~\ref{moda}. 
\end{proof}

Let us note that the above theorem also gives rise to similar results for parity versions of Set Partition, Set Cover, Hitting Set and CNF-Sat via the reductions of~\cite{DBLP:conf/coco/CyganDLMNOPSW12,DBLP:journals/corr/abs-1112-2275} (since the reduction in~\cite[Theorem 4.9]{DBLP:journals/corr/abs-1112-2275} from Subset Sum$/m$ to Set Partition is parsimonious). But all of this was also shown in a stronger form in \cite{2016arXiv160104743W, 2016arXiv160201295B}.

Also~\cite[Theorem 1.6]{DBLP:journals/corr/AustrinKKN15}, in combination with the above theorem gives an $O^*(2^{.499n})$-sized proof that can be probabilistically verified in time $O^*(2^{.499n})$ for \subsetsum{}. However, combining~\cite[Theorem A.1]{2016arXiv160104743W} with the methods from \cite{DBLP:journals/corr/AustrinKKN15} gives a more efficient proof system leading to a proof of size $O^*(2^{.3113n})$ that can be verified in time $O^*(2^{.3113n})$.

\paragraph{Acknowledgements} The author would like to thank Petteri Kaski for bringing the question of whether \subsetsum{} admits a Merlin-Arthur protocol with proof size and verification time $O^*(2^{(0.5-\epsilon)n})$ for $\epsilon >0$ to his attention. The author was not aware of the contents of \cite{2016arXiv160104743W, 2016arXiv160201295B} while observing this note in the beginning of November'16. This work is supported by NWO VENI project 639.021.438, and was done while visiting the Simons institute in the fall of 2015.
\bibliographystyle{plain}
\bibliography{masubsetsum}

\end{document}